\documentclass[11pt,a4paper]{article}

\usepackage[utf8]{inputenc}
\usepackage{amsmath, amssymb, amsthm}
\usepackage{geometry}
\usepackage{hyperref}
\usepackage{booktabs}
\usepackage{graphicx}

\usepackage{tikz}
\usetikzlibrary{calc, shapes, positioning, arrows.meta,
  decorations.pathreplacing, decorations.pathmorphing}
\usepackage{amsmath, amssymb}
\usepackage[ruled,vlined,linesnumbered]{algorithm2e}

\theoremstyle{definition}
\newtheorem{definition}{Definition}[section]
\newtheorem{axiom}{Axiom}

\theoremstyle{plain}
\newtheorem{theorem}{Theorem}[section]
\newtheorem{lemma}[theorem]{Lemma}

\newtheorem{corollary}[theorem]{Corollary}

\theoremstyle{remark}
\newtheorem{remark}{Remark}[section]

\title{CIPS: Maximal Certified Persistence in Cyber-Physical Systems}

\author{Avinash Malik\\ avinash.malik@auckland.ac.nz \\ Department of Electrical, Computer, and Software
  Engineering, \\ University of Auckland, New Zealand}
\date{\today}

\begin{document}

\maketitle

\begin{abstract}
  We introduce the \textit{Theory of Certified Information Persistence
    Systems} (CIPS), a universal mathematical framework for computing
  the maximal certified persistence of information in cyber-physical
  systems (CPS). CIPS provides an axiomatic foundation that separates
  the continuous evolution of state validity from discrete, memoryless
  control interventions. By accommodating digital sampling and execution
  latency through robust set contraction, the framework mathematically
  isolates a system's \textit{maximal certified persistence horizon}---a
  strict theoretical upper bound on safe autonomous operation relative
  to the system's defined metric growth bounds. Our central
  representation theorem proves that CIPS provides a universal
  representation framework: every empirically safe scheduling policy, is
  structurally isomorphic to a conservative surrogate evaluation within
  a canonical CIPS. By dynamically targeting this latency-compensated
  canonical horizon, the framework minimizes conservatism relative to
  the bounding assumptions, achieving an optimal certified scheduling
  policy, minimizing computational and network interventions while
  mathematically guaranteeing continuous physical safety.
\end{abstract}

\section{Introduction}

Modern Cyber-Physical Systems (CPS)~\cite{lee2017introduction}, ranging
from autonomous vehicle platoons to agile unmanned aerial vehicles
(UAVs), operate under strict safety, stability, and performance
requirements~\cite{mellinger2011minimum,zheng2015stability}. To
mathematically guarantee operational correctness, these systems
construct and evaluate formal validity certificates. These range from
continuous Control Barrier Functions (CBFs)~\cite{ames2019control} for
collision avoidance~\cite{wang2017safety,ariu2017chance,jo2026geometry}
and forward reachability sets for safe
navigation~\cite{althoff2014reachability}, to Model Predictive Control
(MPC)~\cite{mayne2000constrained} feasibility bounds. However,
maintaining, computing, or transmitting these real-time guarantees
incurs prohibitive computational, energetic, and bandwidth costs. For
instance, continuous Vehicle-to-Vehicle (V2V) broadcasting congests
shared wireless spectrums, while high-frequency onboard optimization
rapidly drains limited robotic battery reserves.

Because computing or transmitting these validity certificates is
inherently expensive, systems cannot afford continuous evaluation.
Intermittent execution is practically mandatory. The critical
operational challenge then becomes identifying the optimal times to
recompute, transmit, or refresh this information before physical safety
is violated.

Currently, every engineering and computational discipline solves this
intermittent scheduling problem independently. Event-Triggered Control
(ETC)~\cite{tabuada2007event} and Self-Triggered Control
(STC)~\cite{heemels2012introductory} dictate when sensors should sample
or networks should transmit to preserve stability; adaptive Ordinary
Differential Equation (ODE) solvers decide when to recompute numerical
steps based on local error estimates; and distributed computer science
architectures construct bespoke Time-To-Live (TTL) policies to
invalidate stale network cache data~\cite{bernstein1987concurrency}.
Each field defines its own application-specific validity certificate and
derives a bespoke, isolated regeneration policy.

This fragmented landscape naturally prompts a fundamental question:
\textit{Is there an all-encompassing representation and theory of
  maximal certified information persistence?}

Yes, there is. In this paper, we introduce the \textit{Theory of
  Certified Information Persistence Systems} (CIPS), a unified
mathematical framework for computing the maximal certified persistence
of information. Rather than proposing yet another domain-specific
scheduling algorithm, CIPS provides an abstraction that separates
domain-specific certificate construction from domain-independent
persistence reasoning. By defining exactly how long a computed
certificate remains valid, CIPS fundamentally computes $\tau^*(I)$ for some
information object $I$: the greatest safe reuse interval permitted by an
underlying certification under the assumed metric drift bounds.
Everything else in a system's execution follows from this supremum.

The \textbf{contributions} of this paper are structured as follows:
\begin{enumerate}
\item \textbf{A New Mathematical Abstraction:} We define CIPS as a formal 5-tuple, separating domain-specific certificate construction from domain-independent temporal reasoning.
\item \textbf{A Cornerstone Representation Theorem:} We prove that safe
  scheduling policies admit canonical CIPS representations and are
  strictly upper-bounded by the maximal persistence horizon. This
  establishes that periodic scheduling, heuristic event-triggering,
  adaptive scheduling, etc are all strictly conservative approximations
  of the true CIPS horizon.
\item \textbf{A Theory of Maximal Persistence:} We derive optimal,
  latency-aware regeneration policies that follow directly from the
  maximal horizon, allowing systems to safely maximize temporal
  progress.
\item \textbf{Composition of Heterogeneous Certificates:} We introduce a
  compositional framework allowing fundamentally distinct constraints
  (e.g., deterministic kinematics, stochastic noise, and network
  latency, etc) to be combined seamlessly into a single, unified
  persistence horizon.
\item \textbf{Empirical Computability:} Through case studies in
  autonomous V2V platooning and drone Simultaneous Localization and
  Mapping (SLAM) navigation, we demonstrate how existing CPS problems
  map into CIPS, proving that the maximal persistence horizon can be
  computed in practice to aggressively reduce unnecessary computation
  and communication.
\end{enumerate}

\section{A practical CPS application: V2V Platooning and the Engineers
  Implementation}
\label{sec:v2v-platooning-setup}

To ground the CIPS framework in a practical application, we evaluate a
Vehicle-to-Vehicle (V2V) platooning scenario. Figure~\ref{fig:v2v_setup}
illustrates the physical configuration of the system.

\begin{figure}[htpb]
    \centering
    \begin{tikzpicture}[>=Stealth, scale=0.9, every node/.style={transform shape}]
        \draw[thick, black!60] (-2, 0) -- (12, 0);
        \draw[dashed, black!40] (-2, -1) -- (12, -1);
        \draw[thick, black!60] (-2, -2) -- (12, -2);
        
        \filldraw[fill=blue!10, draw=black, thick, rounded corners=2pt] (0, 0.1) rectangle (2.5, 1.1);
        \filldraw[fill=black!80] (0.4, 0.1) circle (0.25);
        \filldraw[fill=black!80] (2.1, 0.1) circle (0.25);
        \node at (1.25, 0.6) {\textbf{Follower}};
        \node[below] at (1.25, -0.2) {$x_T(t)$};
        
        \filldraw[fill=red!10, draw=black, thick, rounded corners=2pt] (6.5, 0.1) rectangle (9.0, 1.1);
        \filldraw[fill=black!80] (6.9, 0.1) circle (0.25);
        \filldraw[fill=black!80] (8.6, 0.1) circle (0.25);
        \node at (7.75, 0.6) {\textbf{Leader}};
        \node[below] at (7.75, -0.2) {$x_L(t)$};
        
        \draw[<->, thick] (2.5, 1.5) -- (6.5, 1.5) node[midway, above] {Gap $d(t)$};
        \draw[dashed] (2.5, 1.1) -- (2.5, 1.7);
        \draw[dashed] (6.5, 1.1) -- (6.5, 1.7);
        
        \draw[->, thick, blue, decorate, decoration={snake, amplitude=1.5mm, segment length=3mm, post length=2mm}] 
            (7.75, 1.3) .. controls (6.0, 2.5) and (3.0, 2.5) .. (1.25, 1.3) 
            node[midway, above, yshift=2mm, text=black] {Triggered V2V Broadcast ($L_{\text{reg}}$)};
            
        \draw[->, thick] (2.5, 0.6) -- (3.5, 0.6) node[right] {$v_T$};
        \draw[->, thick] (9.0, 0.6) -- (10.0, 0.6) node[right] {$v_L$};
    \end{tikzpicture}
    \caption{Physical state space configuration of the V2V platooning
      system.}
    \label{fig:v2v_setup}
\end{figure}
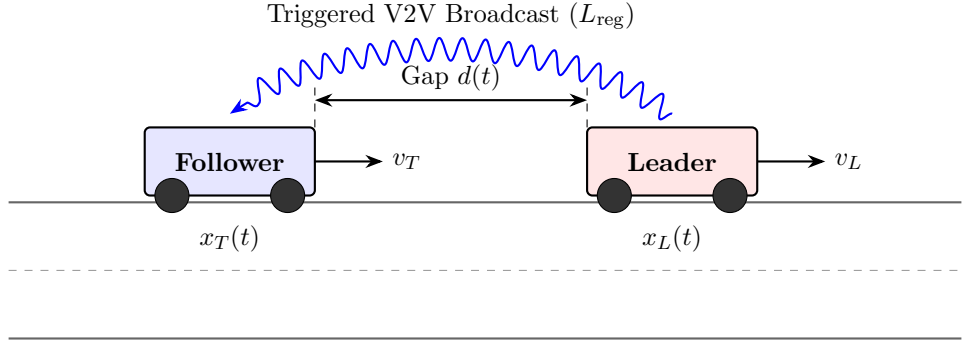

\subsection{Mathematical Problem Formulation}
\label{sec:math-probl-form}

The system state is defined by the relative separation between the
leader and follower vehicles, given by $d(t) = x_L(t) - x_T(t)$. The
\textit{engineering objective} is to transmit packets optimally in order
to maintain a safe gap while subjected to dynamic and computational
constraints. The system parameters are defined as follows:

\begin{itemize}
\item \textbf{Ground-truth safety boundary:} A collision or safety
  violation occurs if the gap falls below $d_{\text{min}} = 3.0$ m. This
  establishes the ground-truth admissible set $\mathcal{A} = \{ d \mid d \geq 3.0 \}$.
\item \textbf{Metric drift rate ($L$):} Under worst-case relative
  braking kinematics, the inter-vehicle gap can close at a maximum rate
  of $L = 8.5$ m/s.
\item \textbf{Computational constraints:} The physical system operates
  in discrete time, meaning the vehicle evaluates and acts upon
  information during zero-order hold blind spots of $\Delta t = 0.30$ s.
  Additionally, when a new V2V broadcast is eventually triggered, the
  network requires a transmission latency of $L_{\text{reg}} = 0.10$ s
  to deliver the packet.
\end{itemize}

Because the follower operates in an open-loop state between the $0.30$ s
updates, the unmitigated trajectory, $d_{\text{unmitigated}}(t)$, must
not be permitted to reach the ground-truth boundary
$\partial \mathcal{A}$. The CIPS framework mathematically guarantees safety by computing
a robust contracted set $\mathcal{A}_{\Delta t}$ and triggering updates ahead of the
latency delay.

\subsection{Algorithmic Implementation}
\label{sec:algor-impl}

\begin{algorithm}[H]
  \caption{Engineer's Workflow: CIPS Formal 5-Tuple Instantiation}
  \label{alg:cips_pipeline}
  \DontPrintSemicolon
  \SetAlgoLined
  \KwIn{Raw vehicle specifications and absolute physical safety constraints}
  \KwOut{Instantiated Ground-Truth CIPS Formal Tuple $\mathfrak{C} = (\mathcal{I}, \mathcal{C}, \mathcal{A}, C, R)$}

  \BlankLine
  \tcp{Step 1: Define Operational State Spaces}
  $\mathcal{I} \leftarrow \text{Follower vehicle trajectory state } x_T(t) \in \mathbb{R}^n$\;
  $\mathcal{C} \leftarrow \text{Relative separation metric space } \mathbb{R}$\;
  \BlankLine

  \tcp{Step 2: Construct Absolute Physical Admissible Set}
  $\mathcal{A} \leftarrow \{ d \in \mathcal{C} \mid d \geq 3.0\,\text{m} \}$ \tcp*{Ground-truth collision threshold $c_{\text{adm}} = 3.0\,\text{m}$}
  \BlankLine

  \tcp{Step 3: Define Metric Mapping and Reset Operator}
  $C \leftarrow \text{Inter-vehicle distance mapping } d(t) = x_L(t) - x_T(t)$\;
  $R \leftarrow \text{Memoryless V2V broadcast and emergency control reset}$\;
  \BlankLine

  \tcp{Step 4: Instantiate Formal Tuple}
  $\mathfrak{C} \leftarrow (\mathcal{I}, \mathcal{C}, \mathcal{A}, C, R)$\;
  \Return $\mathfrak{C}$ \tcp*{Pass pure ground-truth tuple to CIPS runtime execution engine}
\end{algorithm}

The primary advantage of the CIPS framework is the clean separation of
physical system definition from runtime execution. A systems engineer
does not manually program tracking loops, calculate robust safety
margins, or tune trigger conditions. Instead, their sole responsibility
is to map the physical state space, ground-truth safety constraints, and
reset actions into the formal 5-tuple $\mathfrak{C} = (\mathcal{I}, \mathcal{C}, \mathcal{A}, C, R)$.

Algorithm~\ref{alg:cips_pipeline} demonstrates the procedural workflow
the engineer follows to instantiate this ground-truth framework for the
V2V platooning scenario. Once constructed, the tuple $\mathfrak{C}$ is passed
directly to the underlying CIPS runtime engine alongside the hardware
execution parameters ($L, \Delta t, L_{\text{reg}}$). The runtime engine
automatically synthesizes the derived robust admissible region
$\mathcal{A}_{\Delta t}$, compensates for latency, and guarantees physical safety while
achieving \textit{maximal persistence} by delaying network updates as
long as theoretically possible.

In the rest of the paper we describe the formal theory of CIPS along
with the CIPS runtime algorithm to implement the CIPS engine.

\section{CIPS Definition and Axiomatic Foundations}
\label{sec:cips-defin-axiom}

This section starts with giving the fundamental definition of the CIPS
tuple along with the axiom that the CIPS tuple needs to satisfy.

\begin{definition}[Certified Information Persistence System]
\label{def:cips}
A \textit{Certified Information Persistence System} (CIPS) operating under sampling period $\Delta t \ge 0$ and growth bound $L > 0$ is a 5-tuple:
\begin{equation}
\mathfrak{C} = (\mathcal{I}, \mathcal{C}, \mathcal{A}, C, R)
\end{equation}
where:
\begin{enumerate}
    \item $\mathcal{I}$ is a set of \textit{information objects}.
    \item $(\mathcal{C}, d_{\mathcal{C}})$ is a metric space termed the \textit{certificate space}.
    \item $\mathcal{A} \subseteq \mathcal{C}$ is a non-empty closed subset termed the ground-truth \textit{admissible region}.
    \item $C : \mathcal{I} \times \mathbb{R}_{\ge 0} \rightarrow \mathcal{C}$ is an $L$-Lipschitz \textit{certificate mapping}; that is, for all $I \in \mathcal{I}$ and $t_1, t_2 \ge 0$:
    \begin{equation}
    d_{\mathcal{C}}(C(I,t_1), C(I,t_2)) \le L |t_1 - t_2|
    \end{equation}
    \item $R : \mathcal{I} \rightarrow \mathcal{I}$ is a memoryless (Markovian) \textit{regeneration operator}.
\end{enumerate}
\end{definition}

\begin{definition}[Robust Sampled Admissible Region]
  Given a ground-truth admissible region $\mathcal{A}$, growth bound
  $L$, and sampling interval $\Delta t \ge 0$, the \textit{Robust Sampled
    Admissible Region} $\mathcal{A}_{\Delta t}$ is defined via metric set erosion:
  \begin{equation}
    \mathcal{A}_{\Delta t} := \{ x \in \mathcal{A} \mid \mathcal{B}_{L \cdot \Delta t}(x) \subseteq \mathcal{A} \}
  \end{equation}
  where
  $\mathcal{B}_{r}(x) = \{ y \in \mathcal{C} \mid d_{\mathcal{C}}(x, y) \le r \}$ is the closed metric ball
  of radius $r$. In the continuous system ($\Delta t = 0$), the robust region
  recovers the exact ground-truth admissible set: $\mathcal{A}_0 = \mathcal{A}$.
\end{definition}

\begin{remark}[Dynamic and Adaptive Sampling]
\label{rem:varying_deltat}
The sampling interval $\Delta t$ need not remain static across the
operational lifetime of the system. In adaptive sampling algorithms,
multi-rate execution models, or event-triggered control schemes,
$\Delta t$ may vary dynamically at each discrete evaluation instance
$k$, denoted as $\Delta t_k = t_{k+1} - t_k$. Under variable sampling, the
robust region contracts or expands dynamically as
$\mathcal{A}_{\Delta t_k}$, where a smaller sampling step
($\Delta t_k \to 0$) relaxes set erosion toward the true boundary
$\mathcal{A}$. Alternatively, systems operating with non-uniform sampling under a
fixed robust set $\mathcal{A}_{\Delta t}$ remain provably safe by instantiating
$\Delta t$ as the uniform upper bound $\Delta t_{\max} := \sup_k \Delta t_k$.
\end{remark}

The structure $\mathfrak{C}$ is strictly constrained by a fundamental axiom:


\begin{axiom}[Memoryless Regeneration]
  The operator $R$ is strictly memoryless (Markovian). Application of
  $R$ strictly restores state certification into the interior of the
  robust admissible region, entirely independent of the prior trajectory
  history $\{C(I,t) : 0 \le t \le t_k\}$:
  \begin{equation}
    \forall I \in \mathcal{I}, \quad C(R(I), 0) \in \text{int}(\mathcal{A}_{\Delta t})
  \end{equation}
\end{axiom}

\subsection{Derived Definitions}

With the foundational CIPS 5-tuple and its core axioms established, the
framework must bridge the gap between static state definitions and
dynamic temporal execution. The following derived definitions formalize
exactly \textit{when} a discrete system transitions from safe to unsafe,
and mathematically define the maximal duration an information object can
be trusted without requiring regeneration. These constructs form the
critical functional link necessary to evaluate operational persistence
and derive the optimal scheduling theorems presented in Section 4.

\begin{definition}[Sampled Validity State]
  \label{def:validity}
  The sampled validity of an information object $I \in \mathcal{I}$ at an evaluated
  discrete time $t_k = k \cdot \Delta t$ is defined as the Boolean indicator
  function:
  \begin{equation}
    V(I,t_k) = \begin{cases} 
      1, & C(I,t_k) \in \mathcal{A}_{\Delta t} \\ 
      0, & \text{otherwise} 
    \end{cases}
  \end{equation}
\end{definition}

\begin{definition}[Sampled Persistence Horizon]
  The \textit{sampled persistence horizon} $\tau_{\Delta t}(I)$ is the discrete first-exit time corresponding to the earliest sample instance where the certificate mapping evaluates outside the robust admissible region $\mathcal{A}_{\Delta t}$:
  \begin{equation}
    \tau_{\Delta t}(I) = \inf \{ t_k = k \cdot \Delta t \ge 0 : V(I, t_k) = 0 \}
  \end{equation}
  In the exact continuous limit ($\Delta t \to 0$), we denote the
  canonical/exact continuous persistence horizon as:
  \begin{equation}
    \label{eq:1}
    \tau(I) = \inf \{ t \ge 0 : C(I,t) \notin \mathcal{A} \}
  \end{equation}
\end{definition}

\section{Theoretical Framework \& Guarantees}
\label{sec:theor-fram-}

We structure the theoretical properties of CIPS into a five-tiered
hierarchy: fundamental safety guarantees, practical surrogate 
approximations, formal representation theorems, operational 
scheduling optimality, and structural algebraic properties.

\subsection{Fundamental Guarantees (Well-Posedness \& Safety)}
\label{sec:fund-guar-well}

This subsection establishes the core viability and baseline safety of the CIPS framework. Theorem~\ref{thm:boundary} mathematically guarantees that the system perfectly evaluates to the safe boundary at the exact persistence horizon. Lemma~\ref{lem:intersample} provides the critical operational guarantee that the physical system remains safe strictly in between discrete sampling evaluations. Finally, Lemma~\ref{lem:maximality} establishes that no safe scheduling policy can indefinitely defer regeneration beyond this computed theoretical horizon without risking constraint violation.

\begin{theorem}[Boundary Exit Theorem]
  \label{thm:boundary}
  For any CIPS, if the sampled persistence horizon is finite
  ($\tau_{\Delta t}(I) < \infty$), then at the exact first-exit time
  $t_k = \tau_{\Delta t}(I)$, the true certificate resides strictly outside or
  on the boundary of the robust admissible region:
  $C(I, \tau_{\Delta t}(I)) \notin \text{int}(\mathcal{A}_{\Delta t})$. In the continuous exact
  limit ($\Delta t \to 0$), the geometric boundary evaluation is strictly
  recovered: $C(I, \tau(I)) \in \partial \mathcal{A}$.
\end{theorem}

\begin{proof}
  For $\Delta t > 0$, by definition of the infimum defining the first-exit
  time, $C(I, \tau_{\Delta t}(I))$ evaluates outside
  $\mathcal{A}_{\Delta t}$ or on its boundary; if it were in the interior, an open ball
  around the state would exist entirely within the region, contradicting
  the definition of an exit time. As $\Delta t \to 0$, the erosion factor
  $L \cdot \Delta t \to 0$, rendering
  $\mathcal{A}_{\Delta t} \to \mathcal{A}$. Because the certificate mapping
  $t \mapsto C(I,t)$ is $L$-Lipschitz continuous (by
  Definition~\ref{def:cips}), the trajectory is strictly continuous and
  cannot jump over the boundary infinitesimally without intersecting it.
  Since $\mathcal{A}$ is a closed set, this mathematically enforces
  $C(I, \tau(I)) \in \partial \mathcal{A}$. This geometric result links CIPS directly to
  viability theory and barrier certificates.
\end{proof}

\begin{lemma}[Inter-Sample Safety Guarantee]
  \label{lem:intersample}
  Let $C(I,t)$ be $L$-Lipschitz continuous. If $V(I, t_k) = 1$, then
  $C(I, t) \in \mathcal{A}$ for all $t \in [t_k, t_k + \Delta t]$.
\end{lemma}

\begin{proof}
  By $L$-Lipschitz continuity, for any $t \in [t_k, t_k + \Delta t]$, the
  maximal state deviation is bounded by the metric distance
  $d_{\mathcal{C}}(C(I, t), C(I, t_k)) \le L(t - t_k) \le L \Delta t$. The condition
  $V(I, t_k) = 1$ implies $C(I, t_k) \in \mathcal{A}_{\Delta t}$. By the strict metric
  definition of $\mathcal{A}_{\Delta t}$, the closed metric ball
  $\mathcal{B}_{L \cdot \Delta t}(C(I, t_k))$ is entirely contained within
  $\mathcal{A}$. Since the metric distance to $C(I, t)$ is at most
  $L \Delta t$, it follows geometrically that
  $C(I, t) \in \mathcal{B}_{L \cdot \Delta t}(C(I, t_k)) \subseteq \mathcal{A}$, mathematically guaranteeing
  the certificate remains inside the ground-truth set throughout the
  inter-sample interval.
\end{proof}

\begin{lemma}[Maximality Limit]
  \label{lem:maximality}
  Let $\mathfrak{C}$ be a CIPS. No operational scheduling policy evaluating
  $C(I,t_k)$ can safely defer the execution of the regeneration operator
  $R$ beyond $t_k = \tau_{\Delta t}(I)$.
\end{lemma}

\begin{proof}
  By definition, $\tau_{\Delta t}(I)$ is the earliest discrete instance where
  sampled validity fails. Deferring regeneration beyond this point
  guarantees the system evaluates to an invalid state
  ($C(I, \tau_{\Delta t}(I)) \notin \mathcal{A}_{\Delta t}$), strictly violating the preconditions
  of Lemma~\ref{lem:intersample}. Consequently, continuous ground-truth
  validity ($C(I,t) \in \mathcal{A}$) fails, establishing
  $\tau_{\Delta t}(I)$ as the absolute temporal bound for any safe policy.
\end{proof}

\subsection{Practical Realization: Surrogate Approximations}
\label{sec:pract-real-surrogate}

In real-world control systems, computing the exact ground-truth
persistence trajectory ($C$) may be analytically intractable. This
subsection bridges the gap between pure theory and computable
implementation. It formally defines conservative metric dominance and
proves via Lemma~\ref{lem:surrogate_safety} that utilizing an easily
computable, bounding surrogate trajectory is strictly safe. This ensures
systems can execute practical, conservative approximations of the
persistence horizon without violating ground-truth physical constraints.

\begin{definition}[Conservative Dominance]
  For any two states $x, y \in \mathcal{C}$, we say that $x$ conservatively
  dominates $y$ relative to $\mathcal{A}$, denoted as
  $x \succeq_{\mathcal{A}} y$, if the metric distance from $x$ to the exterior
  $\mathcal{A}^c$ is less than or equal to that of $y$:
  \begin{equation}
    \operatorname{dist}(x, \mathcal{A}^c) \le \operatorname{dist}(y, \mathcal{A}^c)
  \end{equation}
  where $\operatorname{dist}(z, \mathcal{A}^c) = \inf_{w \notin \mathcal{A}} d_{\mathcal{C}}(z, w)$. Consequently, for any robust sampled region $\mathcal{A}_{\Delta t}$, $x \in \mathcal{A}_{\Delta t} \implies y \in \mathcal{A}_{\Delta t}$.
\end{definition}

\begin{lemma}[Conservative Surrogate Safety Lemma]
  \label{lem:surrogate_safety}
  Let $\mathfrak{C} = (\mathcal{I}, \mathcal{C}, \mathcal{A}, C, R)$ be a CIPS. Let $C_{\text{upper}} : \mathcal{I} \times \mathbb{R}_{\ge 0} \rightarrow \mathcal{C}$ be a computable surrogate mapping that upper-bounds the true certificate evolution relative to the boundary $\partial \mathcal{A}$:
  \begin{equation}
    C_{\text{upper}}(I,t) \succeq_{\mathcal{A}} C(I,t) \quad \forall t \ge 0
  \end{equation}
  Define the surrogate sampled persistence horizon as:
  \begin{equation}
    \tau_{\Delta t, \text{surrogate}}(I) = \inf \{ t_k = k \cdot \Delta t \ge 0 : C_{\text{upper}}(I, t_k) \notin \mathcal{A}_{\Delta t} \}
  \end{equation}
  Then triggering regeneration $R$ at or before
  $t = \tau_{\Delta t, \text{surrogate}}(I)$ strictly guarantees continuous
  ground-truth validity ($C(I,t) \in \mathcal{A}$) for the true state, and
  $\tau_{\Delta t, \text{surrogate}}(I) \le \tau_{\Delta t}(I) \le \tau(I)$.
\end{lemma}

\begin{proof}
  By the definition of conservative dominance,
  $C_{\text{upper}}(I, t_k) \succeq_{\mathcal{A}} C(I, t_k)$ strictly implies that if
  $C_{\text{upper}}(I, t_k) \in \mathcal{A}_{\Delta t}$, then
  $C(I, t_k) \in \mathcal{A}_{\Delta t}$. Consequently, the surrogate trajectory must
  exit $\mathcal{A}_{\Delta t}$ at or before the true trajectory, immediately enforcing
  $\tau_{\Delta t, \text{surrogate}}(I) \le \tau_{\Delta t}(I)$. Furthermore, for any
  evaluated time $t_k \le \tau_{\Delta t, \text{surrogate}}(I)$, evaluating the
  surrogate guarantees $C(I, t_k) \in \mathcal{A}_{\Delta t}$, which by
  Lemma~\ref{lem:intersample} mathematically ensures
  $C(I,t) \in \mathcal{A}$ throughout the inter-sample interval.
\end{proof}

\subsection{Representation Theorems and Maximal Certified Persistence}
\label{sec:representation-theorems}

This subsection elevates CIPS from a standalone scheduling tool to a
universal covering theory. Theorem~\ref{thm:universal_safety_bound}
proves that \textit{any} heuristically safe scheduling policy (such as
existing Event/Self-Triggered Control methodologies, etc) is
mathematically isomorphic to a canonical CIPS execution operating on a
conservatively contracted admissible sub-level set. Building on this,
Theorem~\ref{thm:converse_cips} demonstrates dimensionality reduction,
proving that arbitrarily complex multi-dimensional policies can always
be canonically compressed into a mathematically equivalent 1-Dimensional
CIPS using standard signed distance functions.

\begin{theorem}[Maximal Safety Bound and Sub-Level Construction]
  \label{thm:universal_safety_bound}
  Let
  $\mathfrak{C}^* = (\mathcal{I}, \mathcal{C}, \mathcal{A}^*, C^*, R)$ be the canonical ground-truth CIPS
  representing the exact physical or logical dynamics of a system, with
  metric certificate space $(\mathcal{C}, d_{\mathcal{C}})$ and non-empty boundary
  $\partial \mathcal{A}^* \neq \emptyset$, yielding the exact continuous persistence horizon
  $\tau^*(I)$. Let $\Pi_{\text{safe}}$ be the set of any safe, causal,
  memoryless scheduling policies $\pi$ that dictate execution times
  $t_k^\pi$ for the regeneration operator $R$. A policy is strictly safe
  if it acts before the ground-truth system fails (i.e., continuous
  ground-truth validity $C^*(I,t) \in \mathcal{A}^*$ is never violated). Then:
  
  \begin{enumerate}
  \item \textbf{Maximal Safety Bound:} For every safe scheduling policy $\pi \in \Pi_{\text{safe}}$, the scheduled intervention time strictly satisfies:
    \begin{equation}
      t_k^\pi \le \tau^*(I)
    \end{equation}
  \item \textbf{Constructive Equivalence:} Let $c_\pi = C^*(I, t_k^\pi)$ denote the exact ground-truth certificate evaluated at the policy's scheduled intervention time. For every safe policy $\pi \in \Pi_{\text{safe}}$, there exists a uniquely constructed contracted admissible region $\mathcal{A}_\pi \subseteq \mathcal{A}^*$ such that evaluating the canonical certificate $C^*$ against $\mathcal{A}_\pi$ perfectly reproduces the policy's intervention schedule. This region is constructed as the boundary clearance sub-level set:
    \begin{equation}
      \mathcal{A}_\pi := \left\{ x \in \mathcal{A}^* \mid \operatorname{dist}(x, \partial \mathcal{A}^*) \ge \operatorname{dist}(c_\pi, \partial \mathcal{A}^*) \right\}
    \end{equation}
  \end{enumerate}
\end{theorem}

\begin{proof}
  \textit{(1. Universal Safety Bound)} Assume for contradiction that
  there exists a safe policy $\pi \in \Pi_{\text{safe}}$ that schedules an
  intervention at elapsed time $t_k^\pi > \tau^*(I)$. By the definition of
  the exact persistence horizon $\tau^*(I)$ (Equation~\eqref{eq:1}), there
  exists some time $t \in (\tau^*(I), t_k^\pi]$ such that
  $C^*(I, t) \notin \mathcal{A}^*$. This strictly violates continuous ground-truth
  validity prior to the regeneration intervention, directly
  contradicting the premise that $\pi \in \Pi_{\text{safe}}$. Hence,
  $t_k^\pi \le \tau^*(I)$ must structurally hold for all safe policies.

  \textit{(2. Conservative Sub-level Representation)} Because
  $(\mathcal{C}, d_{\mathcal{C}})$ is a metric space and
  $\partial \mathcal{A}^* \neq \emptyset$, the distance mapping
  $x \mapsto \operatorname{dist}(x, \partial \mathcal{A}^*) = \inf_{y \in \partial \mathcal{A}^*} d_{\mathcal{C}}(x, y)$ is
  well-defined and 1-Lipschitz continuous. Because
  $t_k^\pi \le \tau^*(I)$, it is guaranteed that
  $C^*(I, t_k^\pi) \in \mathcal{A}^*$. This mathematical necessity ensures
  $\mathcal{A}_\pi$ is a non-empty, closed metric subset of the ground-truth region
  $\mathcal{A}^*$. The exact first-exit time of the true state
  $C^*(I, t)$ from this contracted set $\mathcal{A}_\pi$ evaluates precisely to
  $t_k^\pi$. Thus, any arbitrary safe policy $\pi$ behaves as a strict
  structural isomorphic evaluation of the canonical CIPS operating under
  the conservative threshold $\mathcal{A}_\pi$.
\end{proof}

\begin{corollary}[Maximal Certified Persistence]
  \label{cor:maximal_persistence}
  Let $\mathfrak{C}^*$ denote the canonical CIPS of
  Theorem~\ref{thm:universal_safety_bound} with persistence horizon
  $\tau^*(I)$. Then $\tau^*(I)$ is the unique maximal certified persistence
  horizon among all safe, causal, memoryless scheduling policies.
  Specifically, for every $\pi \in \Pi_{\text{safe}}$,
  \begin{equation}
    t_k^\pi \le \tau^*(I)
  \end{equation}
  Consequently, no safe scheduling policy can exploit a longer certified
  lifetime of the information object than that provided by the canonical
  CIPS\@.
\end{corollary}

\begin{proof}
  The result follows immediately from the Universal Safety Bound
  property established in Theorem~\ref{thm:universal_safety_bound}.
  Since every safe policy satisfies $t_k^\pi \le \tau^*(I)$, no safe policy may
  schedule regeneration after the canonical persistence horizon without
  violating ground-truth validity. Therefore, $\tau^*(I)$ is the maximal
  certified persistence horizon. It is important to note that this
  maximality is strictly conditional on the tightness of the bounding
  assumptions (e.g., the global Lipschitz constant $L$). A less
  conservative local bound $L(x)$ would yield a correspondingly longer
  canonical horizon.
\end{proof}

\begin{remark}[The Boundary of Certification]
  The canonical CIPS is not merely one policy among many. It acts as the absolute boundary separating:
  \begin{itemize}
      \item all mathematically certifiable schedules, and
      \item all schedules that necessarily violate certification.
  \end{itemize}
  This dichotomy grants the canonical persistence horizon $\tau^*$ a
  special theoretical status within the framework, serving as the
  fundamental limit of safe autonomous operation.
\end{remark}

\begin{theorem}[Converse CIPS Dimensionality Reduction]
\label{thm:converse_cips}
Let $(\mathcal{M}, d_{\mathcal{M}})$ be an arbitrary metric state space, and let $\mathcal{S} \subset \mathcal{M}$ be a non-empty closed admissible set with non-empty interior $\operatorname{int}(\mathcal{S}) \neq \emptyset$. Let $\mathcal{I}$ be a set of information objects, where each object $I \in \mathcal{I}$ evolves continuously in the native state space via a trajectory $x_I: [0, \infty) \rightarrow \mathcal{M}$ with initial condition $x_I(0) \in \operatorname{int}(\mathcal{S})$.

Suppose $\pi$ is any object-specific scheduling policy whose intervention time $t_k^\pi(I)$ is dictated by the first exit of the state $x_I(t)$ from $\mathcal{S}$. Then the entire system and policy $\pi$ are exactly generated by a canonical 1-Dimensional CIPS tuple $\mathfrak{C}_\pi = (\mathcal{I}, \mathbb{R}, [0, \infty), C_\pi, R)$, where the certificate mapping $C_\pi$ compresses the native state space into the classical signed distance function $\operatorname{sdist}(x_I(t), \mathcal{S})$:
\begin{equation}
C_\pi(I, t) = \operatorname{sdist}(x_I(t), \mathcal{S}) := \begin{cases} 
+\operatorname{dist}(x_I(t), \mathcal{S}^c), & x_I(t) \in \mathcal{S} \\
-\operatorname{dist}(x_I(t), \mathcal{S}), & x_I(t) \notin \mathcal{S}
\end{cases}
\end{equation}
The continuous persistence horizon of this 1D CIPS perfectly evaluates the original policy schedule: $\tau(I) = t_k^\pi(I)$.
\end{theorem}

\begin{proof}
Since $\mathcal{S}$ is closed and $x_I(t)$ is continuous, the scalar signed distance function $t \mapsto \operatorname{sdist}(x_I(t), \mathcal{S})$ is continuous. By standard properties of signed distance functions on metric spaces:
\begin{enumerate}
    \item $\operatorname{sdist}(x_I(t), \mathcal{S}) > 0 \iff x_I(t) \in \operatorname{int}(\mathcal{S})$,
    \item $\operatorname{sdist}(x_I(t), \mathcal{S}) = 0 \iff x_I(t) \in \partial \mathcal{S}$,
    \item $\operatorname{sdist}(x_I(t), \mathcal{S}) < 0 \iff x_I(t) \in \operatorname{ext}(\mathcal{S}) = \mathcal{S}^c$.
\end{enumerate}
Under the canonical 1D admissible region $\mathcal{A} = [0, \infty)$, the first-exit persistence horizon evaluates precisely to the policy intervention time for that object:
\begin{equation}
  \tau(I) = \inf \{ t \ge 0 : \operatorname{sdist}(x_I(t), \mathcal{S}) < 0 \} = \inf \{ t \ge 0 : x_I(t) \notin \mathcal{S} \} = t_k^\pi(I)
\end{equation}
Because $x_I(0) \in \operatorname{int}(\mathcal{S})$, we have $C_\pi(I, 0) > 0$, ensuring $\tau(I) > 0$ is strictly non-degenerate. Thus, the 1-Dimensional CIPS $\mathfrak{C}_\pi$ canonically represents the native first-exit policy $\pi$ on a per-trajectory basis.
\end{proof}

\subsection{Operational Optimality \& Scheduling}
\label{sec:oper-optim-sched}

Having established how to compute and bound the exact valid duration of states, this subsection dictates exactly \textit{when} a system should act. Theorem~\ref{thm:optimal_policy} leverages the computed persistence horizon to prove that triggering the regeneration operator exactly at the required latency lead-time maximizes operational utility while preserving safety. Additionally, Theorem~\ref{thm:minimal_sequence} demonstrates that chaining these single-step optimal choices sequentially minimizes the total number of resource-heavy interventions across an entire operational timeline. 

\begin{theorem}[Optimal Single-Step Regeneration Policy with Latency]
\label{thm:optimal_policy}
Let $\mathfrak{C}$ be a deterministic CIPS under the following strict operational constraints:
\begin{enumerate}
  \item The certificate evolution $C(I,t)$ contains absolutely no stochastic uncertainty.
    \item Application of $R$ incurs a uniform, constant, strictly positive cost $K > 0$.
    \item Application of $R$ requires a deterministic temporal execution latency $L_{\text{reg}} > 0$.
    \item Operational utility scales linearly with time $t$ between
      regenerations.
    \item $L_{\text{reg}} \le \tau_{\Delta t}(I)$.
\end{enumerate}
Under these constraints, the optimal single-step policy maintaining continuous validity ($V=1$) while maximizing utility is to \textbf{initiate} $R$ strictly at lead time:
\begin{equation}
t_{\text{trigger}} = \tau_{\Delta t}(I) - L_{\text{reg}}
\end{equation}
\end{theorem}

\begin{proof}
Let $t_r$ be the initiation time of the regeneration operator $R$. Because $R$ requires latency $L_{\text{reg}} > 0$ to complete computation/execution, the renewed state takes effect at time $t_{\text{renew}} = t_r + L_{\text{reg}}$. To guarantee continuous ground-truth validity ($V=1$) without state invalidity, renewal must occur on or before the first-exit time: $t_r + L_{\text{reg}} \le \tau_{\Delta t}(I)$, which imposes the safety constraint $t_r \le \tau_{\Delta t}(I) - L_{\text{reg}}$. 

Under uniform execution cost $K > 0$ and linear utility scaling,
maximizing net utility over a single cycle reduces to maximizing the
operational duration prior to initiating regeneration,
$T_{\text{op}} = t_r$. Solving the constrained optimization problem
$\max (t_r)$ subject to $t_r \le \tau_{\Delta t}(I) - L_{\text{reg}}$ yields a
unique maximizer at $t_r^* = \tau_{\Delta t}(I) - L_{\text{reg}}$.
\end{proof}

\begin{remark}[Relationship Between Cost and Latency]
\label{rem:cost_latency}
In computational and physical implementations (e.g., numerical optimization solvers or network queries), execution latency $L_{\text{reg}}$ is often the primary driver of resource consumption. Without loss of generality, operational cost may be modeled as proportional to latency, $K = c \cdot L_{\text{reg}}$ for unit cost rate $c > 0$, or $K = L_{\text{reg}}$ under non-dimensionalized system units.
\end{remark}

\begin{definition}[High-Probability Envelopes]
\label{def:approach_b}
Let $C_{\text{stoch}}(I,t,\omega)$ be a stochastic process (e.g., Brownian
motion with drift, sub-Gaussian error growth) representing certificate
evolution where the exact trajectory cannot be continuously observed in
advance. We construct a deterministic confidence tube
$C_{\text{upper}}(I,t)$ leveraging Lemma~\ref{lem:surrogate_safety},
such that for a given risk tolerance $\epsilon \in (0, 1)$:
$$ \mathbb{P}\left(C_{\text{stoch}}(I,t,\omega) \le C_{\text{upper}}(I,t)\right) \ge 1 - \epsilon $$
\end{definition}

\begin{remark}[Chance-Constrained Surrogate Policy]
\label{rem:chance_constrained}
Because the bounding envelope $C_{\text{upper}}(I,t)$ constructed in
Definition~\ref{def:approach_b} is deterministic, its corresponding
surrogate horizon $\tau_{\Delta t, \text{surrogate}}(I)$ evaluates
deterministically. Therefore, initiating the regeneration operator $R$
at $t = \tau_{\Delta t, \text{surrogate}}(I) - L_{\text{reg}}$ serves as the
optimal single-step policy that mathematically guarantees
$(1-\epsilon)$-chance-constrained continuous
safety~\cite{oguri2019convex}.
\end{remark}

\begin{theorem}[Minimal Regeneration Sequence with Latency]
\label{thm:minimal_sequence}
Let $\mathfrak{C}$ be a CIPS operating under the conditions of Theorem~\ref{thm:optimal_policy}. For any fixed operational timeline $[0, T]$, let $\{t_k\}$ denote a valid sequence of initiation times for $R$. Let $I_k$ denote the information object generated at step $k$, where $I_0 \in \mathcal{I}$ is the initial information object and $I_k = R(I_{k-1})$ for all $k \ge 1$. The sequence of initiation times defined strictly by:
\begin{equation}
t_1 = \tau_{\Delta t}(I_0) - L_{\text{reg}}, \quad \text{and} \quad t_{k+1} = t_k + \tau_{\Delta t}(I_k) \quad \forall k \ge 1
\end{equation}
minimizes the total number $N$ of applications of the regeneration operator $R$ required to cover $[0, T]$ without invalidity gaps.
\end{theorem}

\begin{proof}
We proceed by induction to demonstrate that the greedy maximal lead-time policy maximizes step-wise temporal progress. Let $\{t_k^*\}$ be the optimal initiation sequence generated by $t_1^* = \tau_{\Delta t}(I_0) - L_{\text{reg}}$ and $t_{k+1}^* = t_k^* + \tau_{\Delta t}(I_k)$, and let $\{t_k\}$ be any other valid safe initiation sequence.

\textit{Base Case ($k=1$):} The initial state $I_0$ expires at $\tau_{\Delta t}(I_0)$. Continuous validity requires the first regeneration to complete by $t_1 + L_{\text{reg}} \le \tau_{\Delta t}(I_0)$, giving $t_1 \le \tau_{\Delta t}(I_0) - L_{\text{reg}} = t_1^*$. Thus, $t_1^* \ge t_1$.

\textit{Inductive Step:} Assume $t_k^* \ge t_k$ for some $k \ge 1$. At step
$k$, the renewed object $I_k = R(I_{k-1})$ takes effect at time
$T_k = t_k + L_{\text{reg}}$. By Axiom 1 (Memoryless Regeneration), the
certificate for $I_k$ is restored to
$\text{int}(\mathcal{A}_{\Delta t})$ and remains valid for a duration
$\tau_{\Delta t}(I_k)$, expiring at time
$T_k + \tau_{\Delta t}(I_k) = t_k + L_{\text{reg}} + \tau_{\Delta t}(I_k)$. The
subsequent initiation $t_{k+1}$ must complete by this expiration time:
\begin{equation}
t_{k+1} + L_{\text{reg}} \le t_k + L_{\text{reg}} + \tau_{\Delta t}(I_k) \implies t_{k+1} \le t_k + \tau_{\Delta t}(I_k)
\end{equation}
Applying the inductive hypothesis $t_k^* \ge t_k$, we obtain:
\begin{equation}
t_{k+1}^* = t_k^* + \tau_{\Delta t}(I_k) \ge t_k + \tau_{\Delta t}(I_k) \ge t_{k+1}
\end{equation}
By induction, $t_k^* \ge t_k$ for all $k \ge 1$. Consequently, the sequence $\{t_k^*\}$ covers maximal temporal distance at every step index $k$, requiring the minimal integer $N$ steps such that $t_N^* + L_{\text{reg}} + \tau_{\Delta t}(I_N) \ge T$.
\end{proof}

\subsection{Structural \& Algebraic Properties}
\label{sec:structural-algebraic}

The final theoretical subsection codifies how the framework reacts to complex configurations and multi-constraint environments. Theorem~\ref{thm:fundamental_properties} lays out rules for system composition and constraint modification. It proves that heterogeneous system constraints can be seamlessly combined by taking the logical minimum of their independent persistence horizons, and details the rigorous monotonic behaviors expected when system boundaries are hierarchically bound or expanded.

\begin{theorem}[Fundamental Properties of Persistence]
\label{thm:fundamental_properties}
Let $\mathfrak{C}$ be a CIPS. The framework fundamentally satisfies the following structural properties:
\begin{enumerate}
    \item \textbf{Compositional Horizon:} For the composite system $\mathfrak{C}_{comp} = \mathfrak{C}_1 \cap \mathfrak{C}_2$ constructed by the intersection of their admissibility constraints, the persistence horizon strictly evaluates to $\tau_{comp} = \min(\tau_1, \tau_2)$.
    \item \textbf{Hierarchical Bound:} If an object $I_1$ depends strictly on $I_2$ such that $C_1 \in \mathcal{A}_1 \implies C_2 \in \mathcal{A}_2$, then their persistence horizons are strictly bound: $\tau_{\Delta t}(I_1) \le \tau_{\Delta t}(I_2)$.
    \item \textbf{Monotonicity of Admissibility:} For systems differing only by admissible regions where $\mathcal{A}_1 \subseteq \mathcal{A}_2$, it holds structurally that $\tau_{\mathfrak{C}_1}(I) \le \tau_{\mathfrak{C}_2}(I)$ for all $I \in \mathcal{I}$.
\end{enumerate}
\end{theorem}

\begin{proof}
(1) The composite admissible state demands simultaneous validity, requiring the Boolean conjunction $V_1 = 1 \land V_2 = 1$. The first-exit time defined by the infimum over this logical conjunction simplifies mathematically to the minimum of individual exit times, $\min(\tau_1, \tau_2)$. 
(2) This is established directly via the logical contrapositive: a forced exit from $\mathcal{A}_2$ strictly enforces a simultaneous or prior exit from $\mathcal{A}_1$. The first-exit time of $I_1$ therefore cannot exceed that of $I_2$.
(3) Because $\mathcal{A}_1 \subseteq \mathcal{A}_2$, any geometric point evaluating outside $\mathcal{A}_2$ must implicitly evaluate outside $\mathcal{A}_1$. The exit set for $\mathcal{A}_1$ acts as a superset of the exit set for $\mathcal{A}_2$, guaranteeing the infimum first-exit time for $\mathfrak{C}_1$ occurs at or before that of $\mathfrak{C}_2$.
\end{proof}

\section{Canonical Sampled CIPS Runtime Engine}
\label{sec:cips_engine}

To bridge the abstract mathematical framework with embedded hardware
implementation, this section operationalizes the theory into an
executable runtime architecture.

Algorithm~\ref{alg:cips_runtime_engine} directly operationalizes
Theorem~\ref{thm:optimal_policy} (and its underlying persistence bounds)
by continuously evaluating the abstract state certificate against the
derived robust admissible region, scheduling asynchronous state
regeneration tasks to compensate for execution latency.

\begin{algorithm}[H]
  \SetKwInOut{Input}{Input}
  \SetKwInOut{Output}{Output}
  \caption{Canonical Sampled CIPS Closed-Loop Execution Engine}
  \label{alg:cips_runtime_engine}

  \Input{Instantiated CIPS Tuple $\mathfrak{C} = (\mathcal{I}, \mathcal{C}, \mathcal{A}, C, R)$, Max metric drift bound $L$, Hardware sampling step $\Delta t$, Processing/Network latency $L_{\text{reg}}$}
  \Output{Discrete certification execution schedule $\{t_k\}_{k=1}^K$, Certificate trajectory $C(I(t)) \in \mathcal{C}$}

  \BlankLine
  \tcp{Phase 1: Robust Admissible Region Contraction}
  Construct metric-contracted robust admissible region via metric inner buffer:
  $\mathcal{A}_{\Delta t} \gets \{ c \in \mathcal{A} \mid d_{\mathcal{C}}(c, \mathcal{C} \setminus \mathcal{A}) \ge L \cdot \Delta t \}$\;

  \BlankLine
  \tcp{Phase 2: Closed-Loop Runtime Execution Loop}
  Initialize simulation time $t \gets 0$, trigger counter $k \gets 0$, $t_{\text{last}} \gets -\infty$\;
  Initialize latency queue $\mathcal{Q} \gets \emptyset$ \tcp*{e.g., standard FIFO implementation}

  \While{System Operational}{
    Acquire current system information state: $I \in \mathcal{I}$\;
    Evaluate real-time state certificate: $\gamma(t) \gets C(I, t) \in \mathcal{C}$\;
    
    \BlankLine
    \tcp{Sub-step A: Asynchronous Latency Pipeline Processing}
    \If{$\mathcal{Q} \neq \emptyset$ \textbf{and} $t \ge \text{Head}(\mathcal{Q}).t_{\text{complete}}$}{
      Dequeue completed execution task from $\mathcal{Q}$\;
      Execute state certification reset mapping: $I(t) \gets R(I(t))$\;
    }
    
    \BlankLine
    \tcp{Sub-step B: Closed-Loop Membership Test \& Task Scheduling}
    \If{$\gamma(t) \notin \mathcal{A}_{\Delta t}$ \textbf{and} $(t - t_{\text{last}}) \ge \Delta t$}{
      $k \gets k + 1$\;
      $t_k \gets t$\;
      $t_{\text{last}} \gets t$\;
      Schedule state regeneration task with latency delay: $t_{\text{completion}} \gets t + L_{\text{reg}}$\;
      Enqueue state update task $(t_k, t_{\text{completion}})$ into latency queue $\mathcal{Q}$\;
    }
    
    $t \gets t + \Delta t$\;
  }
\end{algorithm}


\subsection{Operational Phases of the CIPS Engine}
Algorithm~\ref{alg:cips_runtime_engine} executes the formal tuple $\mathfrak{C} = (\mathcal{I}, \mathcal{C}, \mathcal{A}, C, R)$ across two distinct operational phases, fully preserving the metric-space abstraction of CIPS:

\begin{enumerate}
\item \textbf{Robust admissible region contraction (Phase 1):} Prior to
  runtime execution, the engine receives the ground-truth
  physical/logical admissible set $\mathcal{A} \subseteq \mathcal{C}$. Based on the hardware
  sampling interval $\Delta t$ and metric drift bound $L$, it constructs the
  robust admissible region
  $\mathcal{A}_{\Delta t} = \{ c \in \mathcal{A} \mid d_{\mathcal{C}}(c, \mathcal{C} \setminus \mathcal{A}) \ge L \cdot \Delta t \}$. This
  coordinate-free metric erosion guarantees that any unobserved drift
  occurring during a discrete sampling tick $[t, t + \Delta t)$ cannot exit
  the true admissible domain $\mathcal{A}$.
\item \textbf{Closed-loop evaluation and latency pipeline (Phase 2):} At
  each discrete tick $\Delta t$, the engine acquires the current information
  state $I \in \mathcal{I}$, which may represent physical trajectories, estimator
  state distributions, or communication packets, and evaluates the
  metric certificate $\gamma(t) = C(I, t)$. When
  $\gamma(t) \notin \mathcal{A}_{\Delta t}$, a regeneration task is enqueued into the latency
  queue $\mathcal{Q}$ (implemented here as a FIFO structure for simplicity) with
  completion timestamp $t + L_{\text{reg}}$. By
  Theorem~\ref{thm:optimal_policy}, provided the latency assumptions
  hold, the certificate drifts safely through the buffer region
  $\mathcal{A} \setminus \mathcal{A}_{\Delta t}$ without violating the true admissible boundary before
  the reset mapping $R(I(t))$ takes effect.
\end{enumerate}

\subsection{Computational Complexity and Embedded Feasibility}
The execution engine is explicitly designed for resource-constrained cyber-physical hardware:
\begin{itemize}
\item \textbf{Offline initialization:} The set contraction
  $\mathcal{A}_{\Delta t}$ is performed once during system initialization, incurring
  zero runtime overhead during active operation.
\item \textbf{Online runtime overhead:} Per control cycle $\Delta t$, the
  online engine performs exactly one information state acquisition, one
  certificate mapping evaluation $\gamma(t) = C(I, t)$, and one
  set-membership check $\gamma(t) \in \mathcal{A}_{\Delta t}$. Assuming the computational
  complexity of evaluating the certificate mapping $C$ is denoted by
  $T_C$, the overall worst-case online runtime complexity is strictly:
    \begin{equation}
        \mathcal{O}(T_C) \quad \text{per control cycle},
    \end{equation}
    with $\mathcal{O}(1)$ auxiliary queue management operations. 
\end{itemize}

The runtime engine is agnostic to the internal construction of the
certificate mapping $C$. Consequently, Lyapunov functions, barrier
functions, matrix norms, stochastic confidence bounds, optimization
residuals, and application-specific certificates may be substituted
without modifying the execution engine itself. For standard distance
metrics, control barrier functions, or matrix-norm certificates, $T_C$
is $\mathcal{O}(1)$ or $\mathcal{O}(n)$, making the runtime engine exceptionally lightweight
for embedded microcontrollers and real-time operating system (RTOS) task
schedulers.

\section{Domain Instantiations and Case Studies}
\label{sec:doma-inst-case}

We validate the 5-tuple abstraction
$\mathfrak{C} = (\mathcal{I}, \mathcal{C}, \mathcal{A}, C, R)$ across two physical autonomous system
implementations. By explicitly mapping domain-specific mechanics into
the CIPS framework, we demonstrate how disparate systems mathematically
reduce to the same operational surrogate persistence problem. Crucially, we show how existing scheduling paradigms are simply sub-optimal CIPS implementations, justifying that CIPS acts as a universal covering theory for cyber-physical systems.

\subsection{Case Study 1: V2V Platooning under Communication Latency}
Consider a follower vehicle maintaining an inter-vehicle gap
$d(t) = x_L(t) - x_T(t)$ (c.f Section~\ref{sec:v2v-platooning-setup})
behind a leading vehicle that executes hard braking at
$t = 0.6\,\text{s}$. The system operates under a V2V sampling interval
$\Delta t = 0.30\,\text{s}$ and a network/regeneration latency
$L_{\text{reg}} = 0.10\,\text{s}$.

The system maps directly to the \textit{canonical} CIPS tuple as
follows:
\begin{itemize}
\item \textbf{Object ($\mathcal{I}$):} The follower vehicle's physical state and trajectory.
\item \textbf{Certificate Space ($\mathcal{C}$):} The relative separation space $\mathbb{R}$.
\item \textbf{Admissible Set ($\mathcal{A}^*$):} The exact ground-truth safe separation zone $[d_{\text{min}}, \infty)$, where $d_{\text{min}} = 3.0\,\text{m}$.
\item \textbf{Certificate Mapping ($C^*$):} The inter-vehicle gap $d(t)$, bounded by a maximum relative kinematic drift rate $L = 8.5\,\text{m/s}$.
\item \textbf{Regeneration Operator ($R$):} The processing of the V2V wireless packet and subsequent activation of emergency braking.
\end{itemize}

\textbf{Theoretical Validation:} As illustrated in
Figure~\ref{fig:v2v_cips} (b), the ground-truth boundary is eroded to
form the robust contracted boundary
$\partial \mathcal{A}_{\Delta t} = d_{\text{min}} + L \Delta t = 5.55\,\text{m}$ to account for the
zero-order hold interval. By evaluating the exact first-exit sampled
persistence horizon $\tau_{\Delta t}^*$, Theorem~\ref{thm:optimal_policy}
(Optimal Single-Step Regeneration Policy with Latency) mathematically
guarantees safety by computing the latency-compensated broadcast trigger
$t_{\text{trigger}} = \tau_{\Delta t}^* - L_{\text{reg}}$. Furthermore,
Lemma~\ref{lem:intersample} (Inter-Sample Safety Guarantee) strictly
ensures that the vehicle remains within the ground-truth admissible set
$\mathcal{A}^*$ throughout the entirety of the inter-sample interval prior to
regeneration.

\begin{figure}[htbp]
  \centering
  \includegraphics[width=0.6\textwidth]{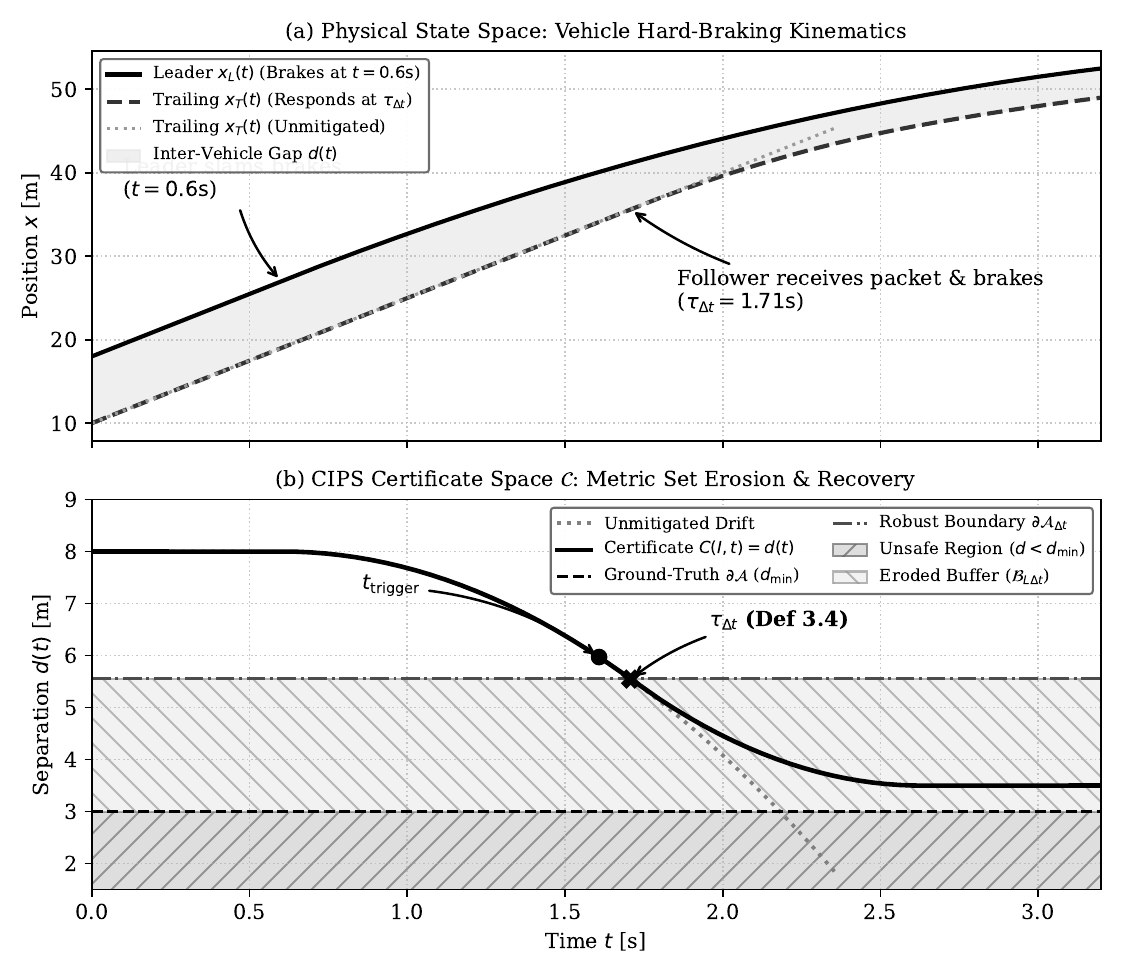}
  \caption{Canonical CIPS applied to 1D V2V Platooning. (a) Physical
    vehicle kinematics during a hard-braking event. (b) The mapped
    certificate space $\mathcal{C}$ demonstrating metric set erosion
    ($\mathcal{B}_{L\Delta t}$) and the precise theoretical trigger evaluated prior to
    boundary violation.}
  \label{fig:v2v_cips}
\end{figure}

\subsection{Case Study 2: 3D UAV Intermittent SLAM Navigation}
\label{sec:case-study-2}

Consider an autonomous quadrotor executing forward flight while
accumulating cross-track drift $(Y(t), Z(t))$. The unmitigated
cross-track error envelope is driven by two competing operational
factors:
\begin{enumerate}
\item \textbf{Deterministic Drift:} A constant translation bias growing
  linearly at rate $L = 0.55\,\text{m/s}$, modeled mathematically as:
  $E_{\text{det}}(t) = L t$.
\item \textbf{Stochastic Sensor Noise:} A high-probability error
  envelope expanding over time governed by the noise coefficient
  $\sigma_Z = 0.95$. This is modeled as a deterministic confidence tube:
  $E_{\text{stoch}}(t) = \sigma_Z \sqrt{t}$.
\end{enumerate}

Leveraging Lemma~\ref{lem:surrogate_safety} (Conservative Surrogate
Safety Lemma) and Definition~\ref{def:approach_b} (High-Probability
Envelopes), the exact continuous certificate mapping evaluates to the
deterministic worst-case composite surrogate error envelope:
$$C_{\text{upper}}^*(I,t) = \max(E_{\text{det}}(t), E_{\text{stoch}}(t))$$

To reset this accumulated drift, the quadrotor must execute an onboard
visual-inertial SLAM routine. Initiating this routine incurs a fixed CPU
wake-up and processing latency of $L_{\text{reg}} = 0.35\,\text{s}$.

The cyber-physical system maps to the \textit{canonical} CIPS tuple as
follows:
\begin{itemize}
\item \textbf{Object ($\mathcal{I}$):} The UAV forward flight trajectory state.
\item \textbf{Certificate Space ($\mathcal{C}$):} The 2D cross-track error plane $\mathbb{R}^2$.
\item \textbf{Admissible Set ($\mathcal{A}^*$):} The outer safety cylinder $\{ (y,z) \mid \sqrt{y^2 + z^2} \le d_{\text{admissible}} \}$, with radius $d_{\text{admissible}} = 1.5\,\text{m}$.
\item \textbf{Certificate Mapping ($C^*$):} The worst-case composite surrogate error envelope $C_{\text{upper}}^*(I,t)$ bounding both translation and noise.
\item \textbf{Regeneration Operator ($R$):} CPU wake-up and complete SLAM localization re-keying to restore state certification to the interior of the robust admissible region (resetting drift to zero).
\end{itemize}

\textbf{Theoretical Validation:} Per
Theorem~\ref{thm:fundamental_properties} (Fundamental Properties of
Persistence), the system's persistence evaluates strictly to the
compositional minimum of the deterministic and stochastic horizons:
$\tau_{\text{comp}}^* = \min(\tau_{\text{det}}, \tau_{\text{stoch}})$. As shown
in Figure~\ref{fig:uav_cips}, triggering CPU wake-up exactly at the
optimal lead time
$t_{\text{trigger}} = \tau_{\text{comp}}^* - L_{\text{reg}}$ safely bounds
the inner robust cylinder. By Theorem~\ref{thm:minimal_sequence}
(Minimal Regeneration Sequence with Latency), sequentially chaining the
memoryless regeneration operator $R$ at this exact latency-aware horizon
prevents the unmitigated trajectory from breaching the outer cylinder
boundary, mathematically guaranteeing crash avoidance over the
operational timeline $[0, 8]$ sec.

\begin{figure}[htbp]
  \centering
  \includegraphics[width=0.8\textwidth]{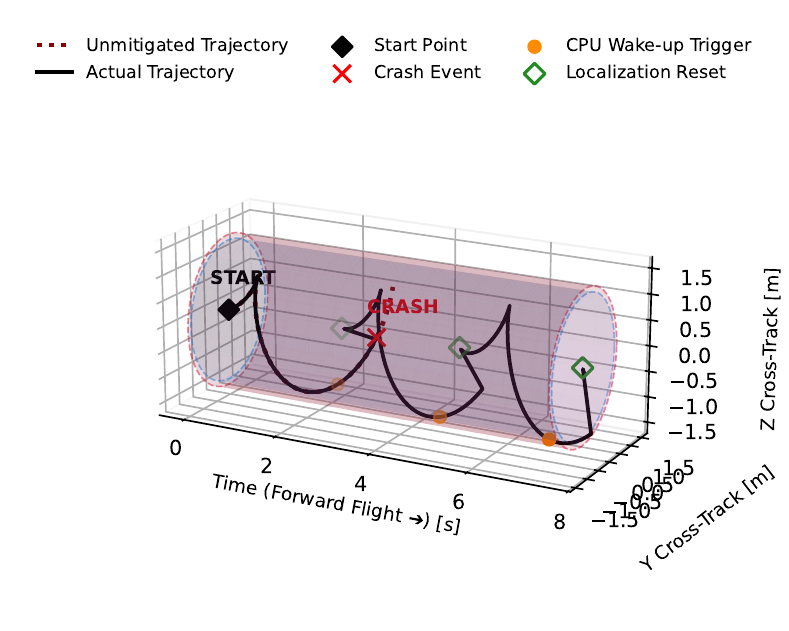}
  \caption{Canonical CIPS applied to 3D UAV Intermittent SLAM. The actual
    trajectory (black) is strictly confined within the ground-truth
    admissible cylinder (red) by applying the memoryless localization
    reset (green diamonds) exactly at the end of the compositional
    persistence horizon.}
  \label{fig:uav_cips}
\end{figure}

\subsection{Experimental Setup and Baseline Implementation}
\label{sec:exper-setup-basel}

To rigorously evaluate the framework, both case studies were evaluated
over a $100\,\text{s}$ operational horizon using four execution
strategies. Crucially, as established in
Theorem~\ref{thm:universal_safety_bound}, all safe scheduling policies
are structurally isomorphic to CIPS; they differ only in how
conservatively they define their admissible sets or schedule
evaluations. We benchmark the exact Canonical closed-loop CIPS against
three distinct execution paradigms:

\begin{enumerate}
\item \textbf{Static Periodic Execution (Degenerate CIPS):} Operates as
  a naive CIPS implementation that heavily contracts the admissible set
  into an artificial, purely temporal bound. It evaluates or transmits
  state certifications at a fixed high sampling frequency
  ($25\,\text{Hz}$ for Case Study 1; $10\,\text{Hz}$ for Case Study 2),
  ignoring the true physical geometry of $\mathcal{A}^*$.
\item \textbf{Heuristic Event-Triggered Control (Sub-Level CIPS):}
  Operates as a conservative CIPS implementation evaluating continuously
  against a static, artificially contracted sub-level set
  $\mathcal{A}_\pi \subset \mathcal{A}^*$. To ensure a fair comparison with zero safety violations
  ($0$ crashes), static safety margins were selected:
  $d_{\text{thresh}} = 6.20\,\text{m}$ for V2V platooning (providing a
  static $3.20\,\text{m}$ buffer above $d_{\text{min}} = 3.0\,\text{m}$)
  and $r_{\text{thresh}} = 1.25\,\text{m}$ for 3D UAV SLAM (a static
  $0.25\,\text{m}$ buffer inside the $1.50\,\text{m}$ safety cylinder).
\item \textbf{Self-Triggered CIPS (Open-Loop CIPS-STC):} Operates as an
  open-loop predictive implementation of CIPS. At each state
  certification reset, it computes the theoretical worst-case maximal
  persistence horizon $\tau^*$ assuming the maximum bounds on drift rate
  $L$ and proactively schedules the next execution at
  $t_{\text{trigger}} = \tau^* - L_{\text{reg}}$, completely eliminating
  intermediate continuous state polling.
\item \textbf{Canonical CIPS (Proposed Closed-Loop CIPS-ETC):} Evaluates
  the exact canonical ground-truth tuple $\mathfrak{C}^*$ against the true
  geometric boundary $\partial \mathcal{A}^*$. It continuously monitors the physical
  state against the latency-compensated boundary and triggers state
  regeneration strictly at
  $t_{\text{trigger}} = \tau^* - L_{\text{reg}}$ per
  Theorem~\ref{thm:optimal_policy}.
\end{enumerate}

We restrict our benchmarking to these static and heuristic policies that
operate under the same global Lipschitz assumptions as the proposed
CIPS. Dynamic ETC policies, which exploit local state-dependent bounds,
are omitted to prevent an asymmetric comparison and are instead deferred
to future work.

\subsection{Empirical Analysis of Benchmark Results}
\label{sec:empir-analys-benchm}

\begin{table}[htbp]
  \centering
  \caption{Quantitative Empirical Benchmark Comparison Across Operational Domains ($100\,\text{s}$ Simulation Horizon)}
  \label{tab:empirical_results}
  \setlength{\tabcolsep}{4pt}
  \resizebox{\linewidth}{!}{
    \begin{tabular}{lcccc}
      \toprule
      \textbf{Domain \& Strategy} & \textbf{Total Interventions} & \textbf{Resource Reduction} & \textbf{Boundary Clearance} & \textbf{Safety Violations} \\
      \midrule
      \textbf{Case Study 1: V2V Platooning} & & & & \\
      Static Periodic (Degenerate CIPS) & 2127 packets & Baseline ($0\%$) & $+4.50\,\text{m}$ & \textbf{0} \\
      Self-Triggered CIPS (CIPS-STC) & 248 packets & $88.34\%$ & $+3.85\,\text{m}$ & \textbf{0} \\
      Heuristic ETC (Sub-Level CIPS) & 25 packets & $98.82\%$ & $+1.80\,\text{m}$ & \textbf{0} \\
      \textbf{Canonical CIPS (Proposed)} & \textbf{15 packets} & \textbf{99.30\%} & \textbf{+1.41\,\text{m}} & \textbf{0} \\
      \midrule
      \textbf{Case Study 2: 3D UAV SLAM} & & & & \\
      Static Periodic (Degenerate CIPS) & 286 wake-ups & Baseline ($0\%$) & $+0.94\,\text{m}$ & \textbf{0} \\
      Heuristic ETC (Sub-Level CIPS) & 48 wake-ups & $83.22\%$ & $+0.13\,\text{m}$ & \textbf{0} \\
      Self-Triggered CIPS (CIPS-STC) & 47 wake-ups & $83.57\%$ & $+0.11\,\text{m}$ & \textbf{0} \\
      \textbf{Canonical CIPS (Proposed)} & \textbf{40 wake-ups} & \textbf{86.01\%} & \textbf{+0.00\,\text{m}} & \textbf{0} \\
      \bottomrule
    \end{tabular}%
  }
\end{table}

The benchmark results in Table~\ref{tab:empirical_results} validate two
core theoretical properties of the framework across continuous physical domains:

1. \textbf{Empirical Validation of Maximal Certified Persistence:} As dictated by the \textit{Maximal Certified Persistence} Corollary~\ref{cor:maximal_persistence}, the Canonical CIPS identifies the absolute upper bound of safe operation time. Because it evaluates directly against the exact ground-truth boundary rather than a sub-level set, its computed horizon $\tau^*$ strictly dominates the operational durations achievable by periodic, heuristic, or open-loop self-triggered policies:
\begin{itemize}
    \item \textbf{V2V Platooning (Case Study 1):} Canonical CIPS initiates packet broadcast precisely when the inter-vehicle gap reaches the robust boundary $\partial \mathcal{A}_{\Delta t} = 5.55\,\text{m}$ (minimum physical gap $4.41\,\text{m}$). Accounting for transmission latency ($L_{\text{reg}} = 0.10\,\text{s}$), control updates arrive exactly in time to arrest vehicle drift, yielding a minimum separation clearance of $+1.41\,\text{m}$ above $d_{\text{min}} = 3.0\,\text{m}$. By contrast, CIPS-STC predicts wake-ups based on the global worst-case relative drift $L = 8.5\,\text{m/s}$. During steady cruise where actual drift is minimal, CIPS-STC is forced to trigger prematurely and repeatedly ($248$ packets, $+3.85\,\text{m}$ clearance) because it cannot observe actual state stability open-loop.
    \item \textbf{3D UAV SLAM (Case Study 2):} Canonical CIPS maintains state certification up to the exact boundary of the outer safety cylinder, achieving $+0.00\,\text{m}$ clearance ($1.50\,\text{m}$ maximum drift) at the instant localization is reset. CIPS-STC operates remarkably close to closed-loop performance in this domain ($47$ wake-ups, $+0.11\,\text{m}$ clearance) because the stochastic/deterministic drift envelope remains strictly active across cycles. However, Canonical CIPS still outperforms STC by observing real-time trajectory state variations rather than relying strictly on the worst-case envelope prediction.
\end{itemize}

2. \textbf{Optimal Resource Allocation:} 
Because Canonical CIPS maximizes the continuous safe temporal span $t_k^\pi \le \tau^*$ at every step using closed-loop boundary observation, it mathematically requires the fewest possible sequential actions to cover a given operational window:
\begin{itemize}
    \item \textbf{V2V Platooning (Case Study 1):} Canonical CIPS eliminates $99.30\%$ of network transmissions compared to continuous periodic execution ($15$ vs. $2127$ packets), outperforms open-loop CIPS-STC by $93.95\%$ ($15$ vs. $248$ packets), and reduces overhead by $40.00\%$ relative to heuristic ETC ($15$ vs. $25$ packets).
    \item \textbf{3D UAV SLAM (Case Study 2):} Canonical CIPS reduces onboard CPU wake-up interventions by $86.01\%$ over continuous periodic sampling ($40$ vs. $286$ wake-ups), while outperforming both heuristic sub-level ETC ($48$ wake-ups) and open-loop CIPS-STC ($47$ wake-ups).
\end{itemize}

These empirical evaluations confirm that evaluating the canonical
persistence horizon on latency-compensated contracted regions yields the
Pareto-optimal trade-off between resource consumption and physical
system viability, while highlighting the clear efficiency advantage of
closed-loop boundary monitoring over open-loop self-triggered
prediction.

\section{Related Work}
\label{sec:related-work}

The problem of maintaining intermittent state validity spans control theory, robotics, information theory, and distributed computer systems. While individual fields have developed specialized techniques for state verification, CIPS provides a domain-agnostic meta-theory by separating certificate construction from certificate persistence.

\subsection{Set Invariance, Viability, and Barrier Functions}
\label{sec:set-invar-viab}

In control theory, safety and state constraints are traditionally
enforced via invariant sets and viability theory. Viability
theory~\cite{aubin2011viability} characterizes the viability kernel: the
set of initial states from which there exists at least one trajectory
remaining indefinitely within constrained bounds. For constrained
dynamical systems, Controlled Invariant Sets~\cite{blanchini1999set} and
Control Barrier Functions (CBFs)~\cite{ames2019control} construct
continuous-time scalar mappings $h(x) \ge 0$ whose superlevel sets define
safe operating regions. Similarly, reachability
analysis~\cite{althoff2014reachability} computes forward reachable sets
to guarantee boundary avoidance.

While CBFs and invariant sets solve the problem of \textit{certificate
construction} for continuous ODEs, they require continuous or
high-frequency evaluation to enforce safety constraints. CIPS does not
replace these barrier functions; rather, it uses them as instances of
the certificate mapping $C(I,t)$, analyzing how long such continuous
certificates remain valid under latency, zero-order holds, and discrete
sampling without requiring continuous re-evaluation.

\subsection{Runtime Verification and Enforcement}
\label{sec:runtime-verification}

Beyond off-line certificate construction, runtime verification (RV)
continuously monitors system trajectories against formal specifications,
often expressed in Signal Temporal Logic
(STL)~\cite{maler2004monitoring}. While RV is traditionally a passive
monitoring process, runtime enforcement (RE) actively intervenes to
preserve system safety. Architectures such as
Simplex~\cite{sha2001using} utilize a highly verified, conservative
safety controller as a fallback for unverified complex controllers,
while shielding techniques~\cite{bloem2015shield} act as a pre- or
post-filter to explicitly block unsafe actions in autonomous and
reinforcement learning environments.

While RV and RE provide critical execution-time safety nets, they typically assume idealized, instantaneous, or fixed-rate discrete monitoring. CIPS complements these paradigms by mathematically bounding the maximum allowable delay between monitor evaluations, effectively dictating \textit{when} a shield or monitor must sample the system to maintain continuous-time safety guarantees.

\subsection{Event-Triggered and Self-Triggered Control}
\label{sec:event-triggered-self}

To reduce computational and communication overhead in cyber-physical
systems, Event-Triggered Control (ETC) and Self-Triggered Control (STC)
substitute periodic execution with state-dependent
updates~\cite{tabuada2007event, heemels2012introductory}. ETC
continuously monitors state-error thresholds to trigger control updates,
whereas STC uses model predictions to proactively compute the next
update time.

CIPS subsumes STC and ETC within a unified metric framework. As proven
in Theorem~\ref{thm:universal_safety_bound}, any safe ETC or STC
triggering policy mathematically reduces to evaluating a canonical CIPS
persistence horizon over a contracted sub-level set $\mathcal{A}_\pi$. Furthermore,
CIPS extends STC principles to non-dynamical domains, such as sensor
wake-up latencies, memory cache invalidation, and asynchronous network
packets.

\subsection{Age of Information (AoI) and Network Caching}
\label{sec:age-information-aoi}

In information theory and network engineering, the Age of Information
(AoI) metric quantifies fresh data timeliness at a destination, defined
as $\Delta(t) = t - U(t)$, where $U(t)$ is the generation timestamp of the
most recent update~\cite{yates2021age}. While AoI measures temporal
freshness linearly, it does not account for state dynamics or safety
boundaries. CIPS generalizes linear age tracking by mapping time decay
through $L$-Lipschitz bounds in a metric state space $\mathcal{C}$, linking
physical drift directly to temporal validity.

In computer science, distributed database materialization, TTL
(Time-To-Live) cache invalidation, and bounded staleness
protocols~\cite{bernstein1987concurrency} preserve system consistency by
bounding information freshness. CIPS provides the mathematical bridge
linking these discrete software caching architectures directly to
physical state propagation in autonomous robotics.

\section{Conclusion and Future Work}
\label{sec:conclusion}

In this paper, we introduced the theory of Certified Information
Persistence Systems (CIPS), a universal mathematical framework designed
to compute and enforce the maximal certified persistence of information
in cyber-physical systems. Moving beyond heuristic or domain-specific
scheduling routines, CIPS provides an axiomatic foundation that formally
separates the continuous evolution of physical state validity from
discrete, memoryless control interventions.

Our central representation theorem establishes that CIPS is not merely a
novel scheduling algorithm, but a foundational covering framework. We
proved that any safe, causal scheduling policy, such as static periodic
execution or event-triggered control, is structurally isomorphic to a
conservative sub-level evaluation within a canonical CIPS. By evaluating
against exact ground-truth geometric boundaries rather than artificial
static thresholds, the canonical CIPS isolates the maximal certified
persistence horizon bounding safe autonomous operation under the
specified metric dynamics.

Furthermore, by explicitly integrating discrete execution latency and
digital sampling through robust set contraction, the framework bridges
the gap between continuous physical dynamics and discrete digital
execution. As demonstrated across diverse autonomous domains,
dynamically targeting this latency-compensated maximal horizon yields
the least conservative certified scheduling policy. Consequently, CIPS
minimizes critical computational and network interventions while
mathematically guaranteeing continuous physical safety, providing a
rigorous, domain-agnostic foundation for the next generation of highly
constrained, safety-critical cyber-physical systems.

While CIPS provides a rigorous foundational meta-theory for persistence
scheduling, transitioning the framework to highly complex, non-linear
cyber-physical systems presents several avenues for future research.

\textbf{State-Dependent Drift Bounds:} The current formulation relies on
a global worst-case $L$-Lipschitz growth bound to guarantee safety
across the entire operational space. In highly non-linear domains,
relying on a global $L$ can introduce conservatism. Future work will
extend the formal 5-tuple to accommodate state-dependent local drift
bounds, $L(x)$, which would dynamically relax the robust set erosion and
yield even longer, less conservative persistence horizons.

\textbf{Computational Tractability of Set Erosion:}
Algorithm~\ref{alg:cips_runtime_engine} relies on the offline geometric
erosion of the admissible set, $\mathcal{A}_{\Delta t}$. While trivial for 1D gaps or
2D cylinders, computing the exact inner metric buffer for arbitrary,
high-dimensional, non-convex constraint sets is challenging. Future
implementations will investigate computational approximations, such as
using sums-of-squares (SOS) programming or neural control barrier
functions (NCBFs), to safely under-approximate the robust region
$\mathcal{A}_{\Delta t}$ without requiring exact geometric erosion.

\textbf{Advanced Stochastic Handling:} Currently, stochasticity is
handled via deterministic high-probability bounding envelopes. While
mathematically sound, this forces the system to schedule based on
worst-case expanding confidence bounds. Integrating true stochastic
reachability analysis into the certificate mapping $C(I,t)$ could allow
CIPS to execute true risk-aware scheduling rather than bounding-tube
approximations.


\bibliographystyle{plain}
\bibliography{ref}

\end{document}